\documentclass[11pt]{article}
\usepackage[margin=1in]{geometry}
\usepackage[utf8]{inputenc}
\usepackage{courier}
\usepackage{graphicx} 
\usepackage{array} 
\usepackage{thmtools}
\usepackage{amsmath,amssymb, amsfonts,amsthm,verbatim}

\usepackage{hyphenat,epsfig,subfigure,multirow}
\usepackage{nicefrac}
\usepackage[usenames,dvipsnames]{xcolor}
\usepackage[boxruled]{algorithm2e}
\SetKw{Continue}{continue}

\usepackage{mdframed}

\usepackage{physics}
\usepackage[T1]{fontenc}
\usepackage{caption}
\usepackage{graphicx}
\usepackage{booktabs}
\usepackage{multirow}
\usepackage{amsmath}
\usepackage{mathtools}
\def\eps{\varepsilon}

\newcommand{\ER}{Erd\"os R\'enyi}
\DeclareMathOperator{\Bin}{Bin}
\DeclarePairedDelimiter\floor{\lfloor}{\rfloor}
\def\prob#1{\mathbb{P}\{#1\}}
\def\Prob#1{\mathbb{P}\{#1\}}

\def\truth{\mbox{Truth}}
\DeclareMathOperator{\Var}{Var}
\usepackage[hidelinks]{hyperref}     
\usepackage[noabbrev,capitalise]{cleveref}

\newtheorem{theorem}{Theorem}
\newtheorem{lemma}{Lemma}
\newtheorem{proposition}{Proposition}

\title{Enabling Asymptotic Truth Learning in a Social Network}
\author{Kevin Lu\thanks{Department of Mathematics, Rutgers University, Email: \texttt{kll160@rutgers.edu}}
\and Jordan Chong\thanks{New York University. Email: \texttt{jhc10007@nyu.edu}} \and
Matt Lu
\thanks{Northwestern University. Email: \texttt{mattlu2029@u.northwestern.edu}}
\and 
Jie Gao
\thanks{Department of Computer Science, Rutgers University, Email: \texttt{jg1555@rutgers.edu}.
Kevin Lu and Gao are supported by NSF IIS-2229876, DMS-2220271, DMS-2311064, CCF-2208663,  CCF-2118953. Part of this work was carried out while JC and ML were participants in the DIMACS REU program at Rutgers University, supported by the NSF award CNS-2150186 and the NSF HDR TRIPODS award CCF-1934924.}}
\begin{document}
\maketitle

\begin{abstract}
Consider a network of agents that all want to guess the correct value of some ground truth state. In a sequential order, each agent makes its decision using a single private signal which has a constant probability of error, as well as observations of actions from its network neighbors earlier in the order. We are interested in enabling \emph{network-wide asymptotic truth learning} -- that in a network of $n$ agents, almost all agents make a correct prediction with probability approaching one as $n$ goes to infinity. In this paper we study both random orderings and  carefully crafted decision orders with respect to the graph topology as well as sufficient or necessary conditions for a graph to support such a good ordering.
We first show that on a sparse graph of average constant degree with a random ordering asymptotic truth learning does not happen. We then show a rather modest sufficient condition to enable asymptotic truth learning. With the help of this condition we characterize graphs generated from the Erd\"os R\'enyi model and preferential attachment model. In an Erd\"os R\'enyi graph, unless the graph is super sparse (with $O(n)$ edges) or super dense (nearly a complete graph), there exists a decision ordering that supports asymptotic truth learning. Similarly, any preferential attachment network with a constant number of edges per node can achieve asymptotic truth learning under a carefully designed ordering but not under either a random ordering nor the arrival order. We also evaluated a variant of the decision ordering on different network topologies and demonstrated clear effectiveness in improving truth learning over random orderings. 
\end{abstract}

\section{Introduction}

Decision making in a social setting is a natural problem that appears in many real world multi-agent settings.
In a security application scenario we have a network of agents making predictions on the state of the network --- is the network under attack and if so which type of attack? All agents can collect local data and statistics and use it to make inferences. In addition, agents can also observe the behaviors, announcements, predictions from neighboring agents in the network. Depending on the particular application scenario, an agent may be able to share with neighboring agents some raw data or even prediction models. In other settings, the agents may belong to different entities. Due to privacy reasons they may not be able to share any data or model beyond their prediction of the network state – a scenario that is considered here. 

Suppose there is an unknown ground truth state $\theta$ and a network of agents that all want to guess the correct value of $\theta$. Each agent can make a single private measurement which has a constant probability to error. 
In addition, an agent can observe the actions/predictions of other agents, who have already made their decisions.
Intuitively an agent can make use of her own private measurements as well as the actions of other agents to improve her chance of a correct prediction. 
This is the setup in the classical sequential learning~\cite{Golub2017-qo,Mobius2014-oy}. The agents make decisions in a given sequential order $\sigma$. The $i$-th agent has the knowledge of the decisions of all the previous agents. The central question asks if the $n$-th agent in the sequence chooses the correct action with probability approaching $1$ as $n$ goes to infinity, phrased as \emph{asymptotic learning}.

Asymptotic learning does not happen in a sequential learning setting -- due to a well known inherent trap called herding or information cascade. An early (small) group of agents make a wrong decision due to observation errors and misfortune, subsequent agents observing the decisions of the early agents ignore their private signal or measurements, in favor of conforming to the actions of the agents they observe~\cite{Banerjee1992-ra,Bikhchandani1992-rs,Welch1992-yt,Smith2000-wk,Chamley2004-or}. This can happen with perfectly rational agents and rigorous Bayesian analysis. When this happens, all agents make the wrong prediction despite an abundant number of independent observations and signals in the network. In reality, we indeed observe herding when a population makes seemingly irrational decisions following the crowd. 

In this paper we consider sequential learning on a social network, i.e., the model used in~\cite{Bahar2020-am}. The agents still make decisions in a sequential order but an agent $v$ can only see the actions of $v$'s neighbors who have arrived earlier. In this case, the network topology limits the visibility (or information flow). Thus both the network topology and the agent ordering are crucial for successful learning. Further, we are interested in enabling \emph{network-wide asymptotic truth learning} -- that in a network of $n$ agents, the average probability of all agents making a correct prediction approaches one as $n$ goes to infinity. Notice that this is a much stronger requirement than asymptotic learning as we need not just the $n$th agent, rather, all but $o(n)$ agents achieve asymptotic learning. 

Enabling network-wide asymptotic learning is much less studied in the literature. One of the main results in~\cite{Bahar2020-am} is to present one family of graphs that supports network wide asymptotic learning, namely, 
the \emph{celebrity} graphs: fully connected bipartite graphs of $m$ celebrity nodes and $n$ normal nodes. Under a random ordering of agents on this network or under an adversarial ordering (but the network topology is randomized), there will be roughly $n/m$ normal nodes that arrive before any celebrity nodes, termed the `guinea pigs'. These early normal nodes are independent of each other and thus make independent decisions. Now, the first celebrity node is able to effectively aggregate information from these independent decisions and arrive at a much more accurate decision. With the Bayesian model, all remaining normal nodes can now effectively learn from this `knowledgeable' celebrity node and thus essentially copy its decision. This leads to highly accurate predictions for all remaining vertices. In this celebrity graph, in order for effective aggregation to happen with high probability, we need $n/m$ to be at least $\Omega(\log n)$ at the first celebrity node. We also need $n/m$ to be sublinear $o(n)$ -- otherwise the network learning rate is dragged down by the guinea pigs. This means that the network is actually very dense, with average degree polynomial in the network size. 

In this paper we are interested in the interplay of network topology and the arrival ordering to enable asymptotic truth learning.
We would investigate sufficient and necessary conditions on whether there is a good ordering and if so present one, and what happens under a random ordering. Our study will also provide insights or guidance on improving quality of learning in real-world decision making processes. 


\subsection{Our Contribution}

Suppose there is a ground truth state $\theta\in \{0, 1\}$ and all agents aim to learn $\theta$. Following a decision order $\sigma$, a node $v$ makes a single private observation which matches $\theta$ with a constant probability $q$. $v$ is also able to observe the decisions of neighbors that are earlier than $v$ in order $\sigma$. 
To make a decision, node $v$ can use the Bayesian model, with all information available so far including the decisions from neighbors, 
the network topology and the ordering $\sigma$. Computing the Bayesian state can be intensive, considering the exponentially growing state space. Alternatively, $v$ can use a simple model such as majority rule on the decisions from neighbors and its private signal.  Our goal is to achieve network wide asymptotic truth learning -- the average success probability of all agents approaches $1$ when the number of agents goes to infinity.



To support asymptotic truth learning, both the network topology and the ordering $\sigma$  matter. If the network is empty, clearly no asymptotic learning is possible -- every agent makes decisions using only local private measurements and thus has a constant error probability. If the network is a complete graph, herding happens with a constant probability. Again we do not have asymptotic truth learning. In the two examples, the ordering $\sigma$ does not matter but in other graphs,  the ordering in which the agents make decisions is crucial. The success of truth learning depends on the combination of a carefully crafted ordering tailored to the network topology, under a proper inference model. 


\smallskip\noindent\textbf{(i) Sparse graph with random ordering: no truth learning} We first start with negative results -- when truth learning does not happen. 
First, when the network has an independent set of size $\Omega(n)$, there is an ordering that does not support truth learning, under both the Bayesian model and the majority rule model. Specifically,  if these independent nodes are put first in the ordering, they have to make independent decisions, using only private observations, leading to a learning rate of $q$ per node. Since this is a constant fraction of the entire network size, the network learning rate is thus bounded away from $1$ by a constant. Using this simple observation we can show that for any network of constant \emph{average} degree, there is an ordering for which truth learning does not happen. In fact, a random ordering does not support truth learning.

\smallskip\noindent\textbf{(ii) A sufficient condition for truth learning in Bayesian model}: Our positive result for enabling truth learning starts from a sufficient condition that concerns both the network structure and the decision ordering. Specifically if there is a vertex $v$ with a set of neighbours $S$, $|S|=\omega(1)$, that are independent of each other such that removing $S$ will disconnect at most $o(n)$ vertices from $v$,  then 
there is a decision ordering such that asymptotic truth learning is achieved under the Bayesian model. 
The decision ordering will put $S$ first before $v$ and then select a decision ordering that propagate the aggregated signal at $v$ to the rest of the vertices. 
The agents in $S$ is analogous to the notion of  \textit{guinea pigs}, agents who make decisions solely based on their private signals~\cite{Smith1991-sy,Sgroi2002-rz}. We show that a slightly super-constant number of such guinea pigs in the neighborhood of agent $v$, allows $v$ to learn a high quality signal -- a signal which matches $
\theta$ with probability approaching $1$, that can be effectively propagated to the entire network as long as connectivity to $v$ is maintained. 

\smallskip\noindent\textbf{(iii) Characterization of {\ER} model and PA model.} Using this sufficient condition we provide characterization of asymptotic truth learning in the {\ER} model -- the classical random graph model with each pair connected by the edge independently decided with probability $p$. When $p$ is too small, $p=O(1/n)$, the graph is too sparse and has a constant fraction of isolated nodes. Thus there is no decision ordering that supports asymptotic truth learning. When $p$ is a constant the graph has quadratically many edges such that any random ordering does not support truth learning since herding happens. Further, when $p$ is almost $1$, with a gap of only $O(n^{\epsilon})/n$ for any $0<\epsilon<1$, one cannot even find an ordering tailored to the specific graph topology to prevent herding. 
For any other value $p$ we show that the sufficient condition is met and therefore a truth learning ordering exists.

We also analyze networks obtained from preferential attachment models. Notice that in the preferential attachment (PA) network each newcomer selects $k$ edges, typically with $k$ being a constant. Thus the network has constant average degree. A random ordering does not support truth learning despite the existence of high degree vertices with the potential power of effectively aggregating information. Further, if we follow the natural arrival order in the PA model, each node knows $k$ vertices earlier in the order. When $k\geq 2$, herding happens with a constant probability and again there is no truth learning. 
By using the sufficient condition and analysis of degree distribution of vertices in a PA model, we show there is indeed a decision ordering in the PA graph that supports asymptotic truth learning. In fact, the preferential attachment network has been considered as one of the `good' networks in the literature that supports truth learning \cite{Bahar2020-am} --- and now we know it must be accompanied by a carefully selected decision ordering. 

\smallskip\noindent\textbf{(iv) Enabling truth learning on sparse graphs.} The sufficient condition requires at least one node with degree $\omega(1)$. On a network of constant degree everywhere, a random ordering does not work. 
Motivated by the negative results, we ask if there can be networks of constant degree that \emph{does} support truth learning with a particular vertex ordering. We first show that on a 2D grid graph we can carefully choose a subtree of size $O(\log n)$ such that the root has a high quality signal. Then we can propagate it to the rest of the network.

Further, we show that the butterfly network with a specific ordering of the vertices, does provide network-wide truth learning. First, we consider a fully balanced binary tree with an ordering bottom up along the tree. We can calculate the success rate of each node in the tree, which is monotonically improved up the tree and approaches $1$ at the root. The butterfly network is essentially packing $n$ fully balanced binary tree with shared vertices at the bottom layers. Again with a bottom up ordering, the learning rate improves along each layer. The network wide learning rate is thus also approaching $1$ as the network size goes to infinity. 
This is the first example of a constant degree network that supports truth learning without using the sufficient condition. 

\smallskip\noindent\textbf{(iv) Simulation results with majority rule on real-world networks.} 
We report experiments using different network topologies including real world graphs and different decision orderings. Since it is computationally expensive to run a full Bayesian model we used the majority rule in all our experiments. We designed a decision ordering inspired by the sufficient condition, which effectively improves truth learning quality over random decision orders. 

\subsection{Related Work}

Learning in a social network is a well studied topic~\cite{Golub2017-qo,Mobius2014-oy}, so we only review work that is most relevant to this study, namely sequential learning with a social network structure. Almost all prior work examine only asymptotic learning, i.e., the $n$-th agent has a success rate approaching $1$ as $n\to \infty$. 
Asymptotic learning occurs
when the private signals are \emph{unbounded} -- when the success probability of a private signal can be arbitrarily close to 1.
When the private signals are bounded (which is the parameter regime studied in this paper), asymptotic learning does not happen due to information cascade~\cite{Smith2000-wk}.

To avoid information cascade one suggestion (\cite{Smith1991-sy,Sgroi2002-rz}) is to restrict social interaction among the agents. In particular the agents early in the sequence, referred to as guinea pigs, are not allowed to observe each other. The late decision makers therefore can infer a great deal from the actions of the guinea pigs. \cite{Sgroi2002-rz} also formulated an optimization problem using a payoff model to find the ideal number of guinea pigs. In this paper, the social network structure is fixed and we are looking for a good decision ordering, if it exists. \cite{Acemoglu2011-vj} also considers a fixed sequence of $n$ agents, but assume that edges in the social network follow a stochastic model. They make a distinction between expanding and non-expanding observations. A social network featuring non-expanding observations describes a topology where the decisions of an infinite number of agents are based on a \emph{finite} set of \textit{excessively influential} agents. Such a network precludes asymptotic learning -- as there is a positive probability that the excessively influential agents made the wrong decision which was copied by the agents later in the sequence. In the positive direction, they show that unbounded private signals and expanding observations are sufficient for asymptotic learning. 
A number of other work use a similar setup with~\cite{Acemoglu2011-vj}.
\cite{Monzon2014-rl} considered position uncertainty and learning robustness. 
\cite{Arieli2018-ap} generalized a linear ordering to an $m$-dimensional lattice. 


\cite{arieli2020social} uses the same model as in~\cite{Bahar2020-am} and provides interesting insights on bounded degree graphs. They pointed out that with a random decision order on a bounded degree graph, the decision process enjoys locality -- the decision of an agent is only influenced by other agents within a certain distance. Further, they also investigated \emph{local learning requirement} that supports the successful learning of a particular agent, without high degree requirements for aggregation. Notice that this result again does not give network-scale asymptotic truth learning.

\section{Models and Problem Definition}
Consider a network as an undirected graph $G=(V,E)$ modeling the social network of $|V| = n$ agents. 
Let $\theta$ be the ground truth signal which is a Bernoulli variable with parameter $\frac{1}{2}$ and is unknown to the agents. 
Each agent makes a private noisy measurement of $\theta$ and, together with observations of the predictions made by neighbours in $G$, try to guess the ground truth. 
Specifically, each agent $v_i$ observes $\theta$ privately with a success probability $q \in (\frac{1}{2}, 1)$, which is a constant away from $1$. We represent the measurement by a Bernoulli random variable $s_i$. $\Prob{s_i=\theta}=q.$
In the literature this is called the \emph{bounded belief} model~\cite{Smith2000-wk,Acemoglu2011-vj}. For simplicity we use the same probability $q$ for every agent. The results can be extended to the case that every agent has an error probability within $(\alpha, \beta)$ with $1/2<\alpha<\beta<1$.

Next, let $\sigma$ be a permutation of $V$, which we call the \emph{decision order}. In the order of $\sigma$, we let each $v_i \in V$ decide on an action, $a_i \in \{0, 1\}$, as a prediction of $\theta$. To simplify notation, we will assume that $v_i$ is the $i$th node in the order $\sigma$, i.e., the decision order is $v_1, v_2, ... v_n$. If two vertices $v_i, v_j$ have an edge in $G$, we write $v_j\sim v_i$. The set of neighbors of $v$ in $G$ is denoted by $N(v)$.
With the decision order $\sigma$ we define the neighbors of a node $v_i$ under $\sigma$ as the set of neighboring nodes that arrive before $i$, i.e., $N_{\sigma}(v_i)=\{v_j|v_j\sim v_i, j<i  \}$. 
When it is $v_i$'s turn to take an action, $v_i$ can use information from its own signal $s_i$ as well as the published actions from nodes in $N_{\sigma}(v_i)$. 

Here are two models that consider different information each agent uses to decide on its action.
\begin{itemize}
    \item \emph{Global model/Bayesian model:} Every agent $v_i$ has the knowledge of the global network topology, the decision ordering $\sigma$ of the vertices, its own private signal $s_i$, as well as the actions taken by its neighbors $N_{\sigma}(v_i)$ that arrive before $v_i$. With knowledge of the network topology, a node uses a Bayesian model to fully utilize the actions of its neighbors in order to find the optimal action. This model is accurate in describing the learning process in the network setting. However, it is also computationally expensive and not practical to implement.

    \item \emph{Local model/majority rule model:} Here, every agent $v_i$ has only knowledge of its private signal $s_i$ and the actions of its neighbors $N_{\sigma}(v_i)$. Without any other knowledge, the agent would just match its action to the majority action of its neighbors $N_{\sigma}(v_i)$ and its own private signal. This model is more suitable for a distributed setting, computationally much more friendly and has been adopted in a number of prior work~\cite{Bahar2020-am,Shoham1992-ir,Laland2004-ej}. 
\end{itemize}

We now define a learning network $(G, q, \sigma)$ with a graph $G$, the observation success rate $q$ and a vertex decision order $\sigma$.
The learning quality for an agent $v_i$ in the learning network $(G, q, \sigma)$ is the probability that its prediction is correct: $\ell(v_i)= \mathbb{P}(a_i=\theta)$. The 
\emph{learning quality} for the entire network $G$ is the ratio of the expected number of agents that learn successfully in the network following the decision order $\sigma$: $$L_{\sigma}(G) = \frac{1}{|V|}E|\{v_i\in V: a_i=\theta\}|=\frac{1}{|V|}\sum_{v_i\in V}\ell(v_i).$$ These definitions follow from the setup in~\cite{Bahar2020-am} and~\cite{arieli2020social}. With these definitions, we can define two new objectives.
Let $\mathcal{F}$ be a family of graphs such that for any $N \in \mathbb{N}$, there exists $G \in \mathcal{F}$ with size at least $N$. We say that $\mathcal{F}$ achieves \emph{network-wide asymptotic truth learning} if as $n\to \infty$, $\inf \{L_{\sigma}(G): |G| \geq n, G \in F\}\rightarrow 1$ with each $L_\sigma(G)$ based on the optimal $\sigma$ for that $G$. We say that $\mathcal{F}$ achieves \emph{network-wide asymptotic random truth learning} if in addition
$\sigma$ is a random permutation on $G$. In future discussions, we can shorten to asymptotic (truth) learning and asymptotic random (truth) learning for brevity. Clearly asymptotic random learning implies asymptotic learning. The goal of this work is to examine both of these objectives. 



\section{Challenges of Asymptotic Truth Learning}\label{sec:challenges}

Before analyzing specific graph models and decision orders that enable truth learning, we first look at topologies and orderings under which truth learning \emph{does not} happen. We start with two cases that are well-known in the literature.
\begin{itemize}
    \item \textbf{Independent learning}: The simplest example in which truth learning cannot be obtained is an empty graph. Each node makes decisions purely from its own noisy observation $s_v$, regardless of the decision order. To maximize the chance of guessing right, by Bayesian rule, the best strategy for an agent $v$ is to report the same value as its private measurement $s_v$, i.e., $a_v=s_v$. The learning quality is $\ell(v) =\prob{s_v=\theta}=q$. The learning quality of the entire network is $q$ as well. Notice that this is not asymptotic truth learning as $L(G)$ is always a constant away from $1$.

    \item \textbf{Information cascade}: When $G$ is a complete graph, each node $v$ knows its own noisy observation $s_v$ as well as all the actions from nodes earlier than $v$ in the order $\sigma$. This is exactly the sequential learning scenario in the literature, regardless of the ordering of nodes. It is well known that information cascade can easily happen in this case~\cite{Bikhchandani1992-rs,Banerjee1992-ra}. In particular, if the first two agents take incorrect actions $a=1-\theta$, then the third agent (and all agents after that) will ignore its own signal and take $a$ as well. This is because the chance that both agents before $v$ are wrong, essentially $(1-q)^2$, is lower than the chance that $s_v$ is wrong ($1-q$). Therefore with probability at least $(1-q)^2$ all agents take the wrong action. This means that the learning quality of the entire network is at most $1-(1-q)^2$, which is again away from $1$ by a constant. 
    Thus asymptotic truth learning does not occur. 
    Notice that the same situation happens if we use the Majority Rule model.
\end{itemize} 

These two examples represent two challenges that forbid asymptotic truth learning. If the graph is too sparse, too many nodes will have to use independent learning, which by itself would already bring down the network learning quality. On the other hand if the graph is too dense such that every node is under heavy influence of the other vertices that arrive before, unfortunate mistakes in observation signals for the first few vertices can propagate throughout the network. This again leads to a large fraction of nodes that end up with the wrong prediction.

\subsection{Graphs of $O(1)$ average degree} 
If a graph is sparse, i.e., has constant average degree, in a random ordering $\sigma$,  a constant fraction of nodes  are independent, i.e., with no neighbors that arrive earlier in the ordering. Therefore they will derive their actions purely from individual observation. Since independent nodes have learning quality to be at least a constant away from $1$, random asymptotic learning does not happen.

\begin{restatable}{theorem}{SparseGraphNoLearning}\label{thm:constant-degree}
    Any family of graphs of constant average degree $\Delta=O(1)$ does not achieve asymptotic random truth learning, under both the Bayesian model and the majority vote model. 
\end{restatable}

The proof is in \Cref{sec:appendix-challenge}. This theorem implies that there is no random asymptotic learning for any constant dimensional grid graphs, constant degree expander graphs, and preferential attachment graph with constant average number of edges per node.

In contrast, \cite{Bahar2020-am} proposed the \emph{celebrity graph} which is a complete bipartite graph with vertex set $X\cup Y$, $|X|=n$, $|Y|=m$. In order to have network learning rate $1-\eps$ under a random decision order, $N=O(\frac{1}{\eps}\ln\frac{1}{\eps})$ and $M=O(\frac{1}{\eps^2}\ln\frac{1}{\eps})$. Therefore to get $|N|+|M|=n$, $\eps$ is in the order of $1/\sqrt{n}$, which means that the network average degree is roughly $\Omega(\sqrt{n})$, in order to support random asymptotic learning.

On the other hand, even for a graph with constant degree everywhere, if we carefully design a decision order it is still possible to achieve asymptotic truth learning. We show such an example called a butterfly network. A \textit{binary butterfly network} consist of $n=(k+1)2^k$ nodes, where $k+1$ is the depth, or \textit{rank}, of the network. Let $(i,j)$ refer to the $j$th node in the $i$th rank. For $i>0$, the node $(i,j)$ forms an edge with $2$ nodes in depth $i-1$: $(i-1, j)$ and $(i-1, m)$, where $m \in \mathbb{Z}$ is found by inverting the $i$ most significant bit of $j$.  Note that due to the  recursive structures, a butterfly network can be interpreted as the integration of $2^k$ complete binary trees.  See Figure~\ref{fig:butterfly} for an example. We highlight one binary tree with the root being the top leftmost vertex. We define a bottom-up ordering $\sigma$ such that vertices of depth $i$ arrive after vertices of depths $i - 1$. The proof is in \Cref{sec:appendix-challenge}.
 
\begin{figure}[htbp]
\centering
\includegraphics[width=0.55\linewidth]{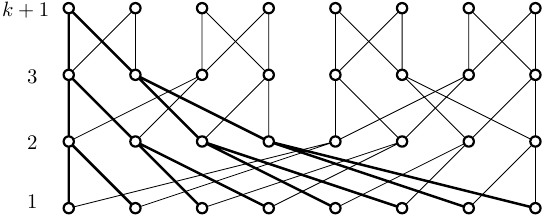}
\caption{An example of a butterfly network of $k=3$.}
\label{fig:butterfly}
\vspace*{-4mm}
\end{figure}

\begin{restatable}{theorem}{ButterFly}
    A class of binary butterfly networks under the bottom-up ordering achieves asymptotic truth learning, under both the Bayesian model and the majority vote model.
\end{restatable}

\section{Enabling Truth learning under Bayesian model}\label{sec:sufficient}


We present the main results in this section. We start with a sufficient condition for asymptotic truth learning under the Bayesian model. Then we use this to analyze {\ER} Graphs, preferential attachment graphs and grid graphs. 

\begin{theorem}\label{theorem:ordering_exists}
Let $\mathcal{F}$ be a family of graphs. For each $G = (V, E) \in \mathcal{F}$ with $|V| = n$, if the following holds for some $v\in V$, $S \subset V$:
\begin{enumerate}
    \item $S \subset N(v)$ such that $S$ is independent ie. $s_1s_2 \notin E$ for all $s_1, s_2 \in S$;
    \item $|S|=\omega(1)$, $|S| = o(n)$;
    \item $G - S$ is a graph where the component containing $v$ has all but $o(n)$ of the remaining vertices,
\end{enumerate}
then there is an ordering $\sigma$ under which $\mathcal{F}$ achieves asymptotic truth learning.
\end{theorem}

\begin{proof}

To prove this theorem, there are two steps. First, we will show that $v$ achieves high learning rate by aggregating the decisions of its independent neighbors. Then, we show that under the Bayesian model, we can propagate this learning rate to the rest of the graph.

\medskip\noindent\textbf{Obtain a high quality signal.} We start with the tail bound of a binomial distribution. Suppose $X$ is a random variable  with a binomial distribution $\Bin(k, q)$ with an integer $k$ and $0
\leq q\leq 1$,  
$$\prob{X=a}={k\choose a}q^a (1-q)^{k-a}.$$
\begin{lemma}[Tail bound of Binomial distribution~\cite{tsun2020}]\label{lem:tail}
Let $X\sim \Bin(n,p)$, and let $\mu=E(X)$. For any $0<\delta<1$, 
$$\prob{X\geq (1+\delta)\mu}\leq \exp(-\frac{\delta^2\mu}{3}),\quad
\prob{X\leq (1-\delta)\mu}\leq \exp(-\frac{\delta^2\mu}{2}).$$ 
\end{lemma}

Let $v$, $S$ be as in the statement of \Cref{theorem:ordering_exists} with $s=|S|$. Then we construct a decision ordering $\sigma$. Let all of the vertices in $S$ decide their actions first, in any order within the set. Since all vertices in $S$ are independent, we have $\prob{a_i = \theta} = q$, $i\in [s]$. Next, we have $v = v_{s+1}$ make its observation and decision. It is clear that $v_{s+1}$'s decision would match the majority between its private observation and the decisions of $S$. Then $\prob{a_{s+1} = \theta} = \prob{X > (s+1)/2}$ where $X \sim \Bin(s + 1, q)$ is a binomial distribution. By~Lemma~\ref{lem:tail}, we have the following probability:
\[\prob{X > \frac{s+1}{2}} = 1 - \prob{X \leq \frac{s+1}{2}}  \geq 1 - \exp{-\frac{(2q-1)^2}{8q}(s+1)}=1 - e^{-O(s)}.\]
Since $s = \omega(1)$, $s\to \infty$ when $n\to \infty$. Thus, 
$\lim_{n \to \infty} \prob{a_{s+1} = \theta} = 1$. That is, the learning quality of agent $v_{s+1}$ is asymptotically approaching $1$.

\medskip\noindent\textbf{Propagate a high quality signal.} We first establish the property of monotonicity in the Bayesian model. That is, a node's success chance is no worse than the neighbors who have arrived earlier. The proof can be found in \Cref{appendix:proof-sufficient}.
\begin{restatable}{proposition}{BayesianRate}\label{proposition:propagate}
Given any learning network $(G, q, \sigma)$ under the Bayesian model, $$\ell(v) \geq \max \{\ell(v_i)| v_i \in N_{\sigma}(v)\}.$$
\end{restatable}

Now, we will use \Cref{proposition:propagate} to propagate the high quality signal from $v_{s+1}$. To do this, let the rest of the decision ordering follow the same ordering as a traversal of $V \setminus S$ starting from $v_{s+1}$, and then the disconnected vertices (if any) follow randomly. Since the rest of the ordering outside of the isolated vertices is a traversal, each vertex that comes after $v_{s+1}$ has a neighbor not in $S$ earlier in the order $\sigma$. Thus, by \Cref{proposition:propagate}, we have all of those vertices achieve learning rate at least $\ell(v_{s+1})$.
Thus the network learning rate $L_{\sigma}(G)$ is:
\[L_{\sigma}(G) \geq \frac{1}{n}(1 - e^{-O(s)})(n-o(n))
= 1 - e^{-O(s)}-\frac{o(n)}{n}. 
\]
Since $s=\omega(1)$, when $n$ goes to infinity, the network learning rate approaches 1.
\end{proof}

We remark that the monotonicity property is unique to the Bayesian model. For the majority rule model a node could be misled by the neighbors not having a high quality signal and thus have a dropping learning rate.
\subsection{{\ER} Graph}

An {\ER} graph~\cite{Erdos1960-hh} $G(n, p)$ is defined by $G = (V, E)$, where $|V| = n$, and 
$p \in (0,1)$, every unordered pair of vertices $v_1, v_2 \in V \times V$ is connected by an edge with probability $p$. 
The expected number of edges in an {\ER} graph is $p{n \choose 2}$. 
\cite{Erdos1960-hh} analyzed graph connectivity with different values of $p$.
\begin{itemize}
    \item If $p < 1/n$, then a graph in $G(n, p)$ will almost surely have no connected components of size larger than $O(\log n )$.
    \item If $p = 1/n$, then a graph in $G(n, p)$ will almost surely have a largest component whose size is of order $n^{2/3}$.
    \item If $np \to c > 1$, where $c$ is a constant, then a graph in $G(n, p)$ will almost surely have a unique giant component containing a constant fraction of the vertices. No other component contains more than $O(\log n)$ vertices.
    \item If $p< \frac{(1-\eps)\ln n} {n}$ for $0<\eps<1$, then a graph in $G(n, p)$ will almost surely contain isolated vertices, and thus is disconnected.
    \item If $p> \frac{(1+\eps)\ln n}{n}$ for $0<\eps<1$, then a graph in $G(n, p)$ will almost surely be connected.
\end{itemize}

When discussing {\ER} graphs, we consider the families $\mathcal{F} = \{G(n, p(n))\}$ where $p(n)$ is some function of $n$. In the following discussion, we may use $p = p(n)$ as shorthand. Since {\ER} graphs are random graphs, when considering asymptotic truth learning, we need to define its learning rate differently. Specifically, since an empty graph and a complete each has a non-zero probability of appearing in an {\ER} model, we cannot hope to find a good ordering for each graph in the family. Instead, we take the expected optimal learning rate for each graph $G$ by the probability that $G$ appears in the {\ER} model. Let $F(G)$ be the probability of graph $G \in G_n$, $G_n$ being the set of all graphs of size $n$, occurring according to the probability distribution of $G(n, p(n))$. We define:
\[L_\sigma(G(n,p(n))) = \sum_{G \in G_n} F(G)L_{\sigma}(G).\]
For asymptotic truth learning, we always consider the optimal $\sigma(G)$ for each $G$. For random asymptotic truth learning, $\sigma$ is a random ordering. We will now show various threshold probabilities where the graph achieves or doesn't achieve truth learning with probability approaching 1 as $n$ approaches infinity. Due to space limit, almost all proofs are in \Cref{appendix:proof-sufficient}.

\subsubsection{$p(n) = O(1/n)$: no asymptotic truth learning}

When $p$ is in the order of $O(1/n)$, there are a constant fraction of isolated vertices, who have a learning rate of $q$ in any ordering $\sigma$. Thus there is no ordering that supports truth learning. 

\begin{restatable}{proposition}{Erdoslowdegree}\label{proposition:ER-low}
The family of {\ER} graphs with parameter $p(n) = \frac{O(1)}{n}$ does not achieve asymptotic truth learning with any ordering $\sigma$ under both the Bayesian or majority rule model.
\end{restatable}

\subsubsection{$p(n) = \frac{o(n)}{n}$ and $p(n) = \frac{\omega(1)}{n}$: asymptotic truth learning}

\begin{proposition}\label{prop:erdos-learning}
A family of {\ER} graphs with parameter $p$ satisfying $p = \frac{o(n)}{n}$ and $p = \frac{\omega(1)}{n}$ achieves asymptotic truth learning under the Bayesian model.
\end{proposition}
To prove this proposition, we will first prove the following lemma (with proof in \Cref{appendix:proof-sufficient}) and then use the sufficient condition of Theorem~\ref{theorem:ordering_exists}. 

\begin{restatable}{lemma}{ErdosLargeComponent}\label{lemma:ER-large-component}
An {\ER} graph $G(n, p)$ with $p = \frac{o(n)}{n}$ and $p = \frac{\omega(1)}{n}$ almost surely contains a giant component of size at least $n - o(n)$. 
\end{restatable}

\begin{proof}[Proof of \Cref{prop:erdos-learning}]
Since an edge exists between any two vertices with independent probability $p$, the degree of each vertex follows the binomial distribution, $\Bin(n-1, p)$, with expectation $d = p(n-1)$, $d=o(n)$ and $d=\omega(1)$. By Lemma~\ref{lem:tail}, for a constant $\delta > 0$, $\prob{\deg(v_i) \notin ((1-\delta)d, (1+\delta)d)} \leq O(e^{-\delta^2d})$. 
Since $d = w(1)$, a randomly selected vertex $v_{\alpha} \in V$ has degree between $((1-\delta)d, (1+\delta)d)$ with high probability. Next, we select a random subset $S$ of the neighbors of $v_\alpha$ such that $|S| =s= 1/p^{1/3}$ if $p \geq 2/n^{3/4}$ and $d/2$ if not. In the first case, the probability of $S$ being an independent set is  $(1-p)^{s \choose 2} \geq 1-p \cdot {s \choose 2} \geq 1-p^{1/3}$.  In the second case, the probability is at least $1-\frac{pd^2}{4} \geq 1-1/n^{1/4}$. In both cases, $s=o(n)$ and $s=\omega(1)$. Thus, with high probability there exists an independent subset of neighbors of $v_\alpha$ of size $\omega(1)$.

Next, we need to show that $G$ remains connected after removing $v$ and $S$. It is clear that removing a random vertex $v$ from $G(n, p)$ is just $G(n-1,p)$. 
Removing all of the neighbors of $v$ is equivalent to removing nodes with probability $p$ from $G(n-1, p)$, so the resulting graph is a random graph $G(k, p)$, where $k$ takes values within $0$ and $n-1$ with probability following distribution $\Bin(n-1, 1-p)$. We call this resulting graph, $G'$. From Lemma~\ref{lem:tail} with high probability, $k = \Theta(n)$. 

With our parameter range of $p$, $G'$ contains with high probability a large component of size $n - o(n)$ by Lemma~\ref{lemma:ER-large-component}. Now in all cases, we know that $|R| = |N(v_\alpha)\setminus S| \geq (1/2-\delta)d = \omega(1)$, so we just need $R$ to be connected to all but $o(n)$ vertices of the graph. It is now trivial to see that the probability of at least one edge from $R$ reaching the large component approaches 1, thus satisfying the final necessary condition of Theorem~\ref{theorem:ordering_exists}.
\end{proof}

\subsubsection{$p(n) = \Theta(1)$: no random asymptotic truth learning}
When $p=\Theta(1)$, we have a dense graph with a quadratic number of edges. If we fix an ordering $\sigma$ of the vertices, we argue that herding happens on a graph generated randomly from {\ER} model. Thus, if we do not tailor the ordering to the specific graph topology obtained from the {\ER} model we do not have truth learning.

\begin{restatable}{proposition}{ErdosRenyiConstantP}\label{prop:ER-constant}
    A family of {\ER} graphs with $p = \Theta(1)$, does not achieve random asymptotic random truth learning for the Bayesian or majority vote models.
\end{restatable}

The above result does not eliminate the possibility of a good ordering chosen with respect to the particular graph topology.  Below we present two results depending on how far $p$ is from $1$. If $p$ is away from $1$ by  $\frac{\omega(n^\epsilon)}{n}$ truth learning exists. Otherwise, if $p$ is away from $1$ by $\frac{O(n^\epsilon)}{n}$, the graph is nearly a complete graph and there is no ordering for asymptotic truth learning. 

\subsubsection{$1-p(n) =\frac{\omega(n^\epsilon)}{n}$: asymptotic truth learning}

In this scenario, we argue that the conditions of \Cref{theorem:ordering_exists} are met. That is, since $p$ is away from $1$ by $\frac{\omega(n^\epsilon)}{n}$, there is an independent set of size $\omega(1)$ in the network. Thus we can choose an ordering that starts from these independent nodes and then propagate the aggregated signal to the rest of the network.   

\begin{restatable}{proposition}{ErdosRenyiComplementGreater}\label{prop:ER-complement-greater}
A family of {\ER} graphs with parameter $p$ satisfying $1-p = \frac{\omega(n^\epsilon)}{n}$ for any $1>\epsilon > 0$ achieves asymptotic truth learning under the Bayesian model.
\end{restatable}

\subsubsection{$1-p(n) = \frac{O(n^\epsilon)}{n}$: no asymptotic truth learning}

When $p$ is too close to $1$, i.e., with gap within $\frac{O(n^\epsilon)}{n}$ for any $1>\epsilon>0$, in any ordering $\sigma$, all but a constant number of nodes have at least two neighbors that are placed earlier in $\sigma$, therefore herding happens with a constant probability which means no truth learning. 

\begin{restatable}{proposition}{ErdosRenyiComplementLesser}\label{prop:ER-complement-lesser}
A family of {\ER} graphs with parameter $p$ satisfying $1-p = \frac{O(n^\epsilon)}{n}$ for any $\epsilon > 0$ does not achieve asymptotic truth learning under the Bayesian model.
\end{restatable}

To prove this, we begin with a lemma:
\begin{lemma}\label{lemma:small-forest}
A family of {\ER} graphs with parameter $p$ satisfying $p = O(n^{-\epsilon})$ for any constant $1 > \epsilon > 0$ has with high probability $O(1)$ maximum induced forest size.
\end{lemma}
\begin{proof}
First, let $N^k_n$ be the random variable for number of $k$-cliques in $G(n, p)$. We have:
\[\mathbb{E}(N^k_n) \leq \binom{n}{k}(1-p)^{\binom{k}{2}-k+1}p^{k-1}(k+1)^{k-1}\]
from the well-known fact that there can be $(x+1)^{x-1}$ possible rooted forests on $x$ labeled nodes. Plugging in $1-p = O(n^{-\epsilon})$, we get:
\[\mathbb{E}(N^k_n) = O(n^{k}e^{-k}n^{-k^2}k^k)\]
It's clear that if $k = \omega(1)$, we get $\mathbb{E}(N^k_n) = o(1)$, which means $\mathbb{P}(N^k_n \geq 1)= o(1)$ from Markov's inequality and thus with high probability, the maximum induced forest size is not $\omega(1)$.
\end{proof}
\begin{proof}[Proof of \Cref{prop:ER-complement-lesser}]
Let $\sigma$ be any decision ordering of $G(n, p)$. Then, let $S_\sigma$ be the set of vertices such that no more than 1 of their neighbors come before them in $\sigma$. We have that $S_\sigma$ must be an induced forest, so $|S_\sigma| = O(1)$ with high probability by \Cref{lemma:small-forest}. Now, we argue that if every node in $S_\sigma$ chooses the wrong decision, then every other node will follow suit regardless of their private signals.

Let's look at the smallest $i$ such that $v_i \notin S_\sigma$. Note that $i \geq 3$. Clearly, every node before it chose incorrectly and  $v_i$ is connected to at least 2 of them, so it will just follow the other nodes' decision. Now this is true inductively. For every node not in $S_\sigma$, it is connected to at least 2 nodes before it, and both of those nodes chose incorrectly, so the current node will also choose incorrectly. Finally, since $|S_\sigma| = O(1)$ the probability of everything in $S_\sigma$ choosing incorrectly is also $O(1)$. Thus, we have herding occurs and asymptotic truth learning does not occur.
\end{proof}

\subsection{Preferential Attachment Graph}

 A preferential attachment (PA) graph, with positive integer $k = O(1)$ as a parameter, is defined by $G = (V, E)$ where $|V| = n$ and the following process for generating edges:
\begin{enumerate}
    \item 
    Start with $G_{k+1}$ 
    as a complete graph on $k+1$ vertices.
    \item Given $G_t$, $t\geq k+1$, generate $G_{t+1}$ by adding a vertex $v$ as well as $k$ undirected edges from $v$ to vertices in $G_t$. Each edge connects $v$ to a vertex $u$ chosen randomly with probability $\frac{\deg(u)}{D}$ where $D$ is the total current degree in $G_t$. 
    Continue until $G_n$ is generated and let $G = G_n$.
\end{enumerate}

We use this variation of the PA model to enforce that the graph is connected. Some other models allow self loops. We note that as $n \to \infty$, this variation converges to the same degree distribution as the ones used in proofs in \cite{Hofstad_2016}, as noted in that book. When $k$ is 1, the graph is a tree. 
When $k\geq 2$, by using the same argument as in Proposition~\ref{prop:ER-constant}, if the decision ordering is the natural arriving order of the PA graph, there is a constant probability that herding happens. Further, for any constant $k$, a random ordering does not support truth learning since the average graph degree is a constant, by \Cref{thm:constant-degree}.
In the following 
we show that for $k\geq 2$ we can achieve asymptotic truth learning by carefully choosing an ordering for graphs following the PA model.

\subsubsection{$k=O(1)$: asymptotic truth learning}
We first introduce a known theorem about the degree of a vertex as the PA model evolves. Define $\Gamma$ as the commonly used \emph{Gamma function}:
\begin{equation}
    \Gamma(t) = \int_{0}^\infty x^{t-1}e^{-x}dx
\end{equation}

\begin{theorem}[\cite{Hofstad_2016}, Exercise 8.14, Equation 8.3.11]\label{theorem:degree}
    Let $v$ be the vertex inserted to make $G_b$. Then as $n \to \infty$ and $a = \Theta(n)$, the degree of $v$, $\deg(v)$ in $G_a$ converges almost surely to a random variable with expected value:
    \begin{equation}
        \sum_{s=1}^k \frac{\Gamma(ka + \frac{1}{2})}{\Gamma(ka)}\frac{\Gamma(k(b-1) + s)}{\Gamma(k(b-1) + s + \frac{1}{2})}
    \end{equation}
\end{theorem}

\begin{theorem}
A class of preferential attachment graphs for any constant $k$ achieves asymptotic truth learning under the Bayesian model.
\end{theorem}

\begin{proof}
Our goal is to show that with high probability, the PA graph satisfies the conditions of \Cref{theorem:ordering_exists}. First, we know that as $n\to \infty$, the degree distribution of $G$ converges almost surely to the power law distribution with exponent $3$ ~\cite{Hofstad_2016}. We may write this distribution as $d(x) = Cx^{-3}$ where $d(x)$ is the proportion of vertices with degree $x$ and $C > 0$ is some constant dependent on $k$. We have with high probability, the number of vertices of degree exactly $k$ converges to $d(k) = Ck^{-3}n = \Theta(n)$. Let $A$ be the set of vertices in $G_t$ with $t=\log(n)$ and let $0 < \epsilon < Ck^{-3}$. \Cref{theorem:degree} says for $a \geq \epsilon n$ and for all $v \in A$, $\deg(v)$ in $G_a$ converges almost surely to a random variable, $\xi$ with expectation:
\begin{equation}
    \mathbb{E}(\xi) = \sum_{s=1}^k \frac{\Gamma(ka + \frac{1}{2})}{\Gamma(ka)}\frac{\Gamma(k(b-1) + s)}{\Gamma(k(b-1) + s + \frac{1}{2})}
\end{equation}
for some $b \leq \log(n)$. A commonly known fact is $\frac{\Gamma(x + \frac{1}{2})}{\Gamma(x)} = \sqrt{x}(1+O(\frac{1}{x}))$ for $x \to \infty$. Then, taking $n \to \infty$ as well as using $k = O(1)$ gives us the complexity $\mathbb{E}(\xi) = \Omega(\sqrt{\frac{a}{\log(n)}})$. That means $\mathbb{P}(\xi = \Omega(\sqrt{\frac{a}{\log(n)}})) > 0$. Then, as $n \to \infty$, we have with high probability that we can find a vertex $v_\alpha \in A$ such that at $G_{\epsilon n}$, $\deg(v_\alpha) = \Omega(\sqrt{\frac{\epsilon n}{\log(n)}})$. We get that the probability for a new vertex inserted at $G_a$ to be a neighbor of $v_\alpha$ is $\Omega(\sqrt{\frac{\epsilon n}{\log(n)}})/(kn) = \Omega(\frac{1}{\sqrt{n \log(n)}})$. Now by definition of $\epsilon$, there are $\Theta(n)$ vertices of degree $k$ in $G$ that were inserted after $G_{\epsilon n}$. Call the vertices in this set that are also connected to $v_\alpha$, $S$. We know $\mathbb{E}(|S|) = \Omega(\frac{1}{\sqrt{n \log(n)}}) \cdot \Theta(n) = \Omega(\sqrt{\frac{n}{\log(n)}})$. We also see that since every vertex in $S$ has degree exactly $k$, their only neighbors are the vertices they connected to upon insertion, so $S$ is independent. For the same reason, $G-S$ is connected, so we have that $v_\alpha$ and $S$ satisfy the necessary conditions of \Cref{theorem:ordering_exists}.
\end{proof}

\subsection{Grid Graph}

The last graph we analyze is the grid graph, i.e., the graph of the lattice points $(i, j)$ with $1\leq i, j\leq k$ with $n=k^2$ and edges connecting two lattice points $(i, j)$ and $(i', j')$ with $|i-i'|+|j-j'|=1$. A grid graph has degree at most $4$ everywhere. Thus there is no random asymptotic truth learning. Below we show that there is a good ordering with which truth learning happens. The main idea is to find an ordering that builds a high quality signal using $o(n)$ vertices, and then propagate this signal to the rest of the network. Since the grid graph has constant degree we cannot use the first phase in our sufficiency condition (\Cref{theorem:ordering_exists}), and have to build that specifically to the grid graph. 

\begin{restatable}{theorem}{GridGraphLearning}\label{thm:gridgraph}
The class of square grid graphs achieves asymptotic truth learning under the Bayesian model.
\end{restatable}

\section{Simulations}\label{sec:short-simulation}
We will now provide simulated results with different network topologies and orderings. 
Since the Bayesian model is computationally intensive, in all of our simulations, we utilize the Majority Rule model. Furthermore finding a maximum independent set is NP-hard so we use heuristic ideas to generate an ordering. 
Note that for each run of the simulation on a single network, a new ordering is created, but the underlying graph structure remains the same. 
We repeat the simulation on the same network and take the average network learning rate over multiple runs. Finally, in every simulation we utilize a private signal $q = 0.7$. 

\smallskip\noindent\textbf{Ordering}
We use random ordering as a baseline and compare with an ordering tailored to the network structure.
We first introduce a simple \emph{Two Neighbors} ordering: 
  Choose a node $v$ at random.
  With an arbitrary order, the nodes in $N(v)$ make their decisions based on the majority rule.
  After that, $v$ makes a decision.
  Then nodes with at least two neighbors who have already made decisions will then follow.
  Repeat the above until all nodes are included.

Furthermore, we introduce an improved \emph{Two Neighbor} ordering to incorporate ``high-value'' agents.
  Specifically, we pick the two highest-degree nodes $u$ and $v$.
  Using a greedy algorithm, select a size $m<\deg(u)$ independent set of $u$'s neighbors.
  Have each of the nodes in the subset make a decision and then let $u$ make a decision.
  Label $u$ as a high-value agent.
Do the same for node $v$.
  Take all the nodes $S'$ in the network that has at least two neighbors who have made decisions, where the majority of these neighbors are labeled high-value.
  Have all nodes in $S'$ make a decision.
  Repeat until termination. If there are no vertices which have at least two high value neighbors, we take a random order of the remaining vertices.
Note that the size of the $m$ independent set of neighbors should be large enough to minimize the chance of a negative information cascade, while small enough to allow for further aggregation down the line. 

\begin{table*}[h]
\caption{email-univ Network}
\label{tab:6}
\begin{small}
\begin{minipage}{\columnwidth}
\centering
\begin{tabular}{lccc} \toprule
\multicolumn{4}{l}{Simulation results $q=0.7$, $n=1133$, iterations$=300$} \\ \cmidrule(l){1-4}
\multicolumn{1}{l}{\multirow{2}{*}{Ordering}} & \multicolumn{3}{c}{Average learning rate} \\ \cmidrule(l){2-4}
\multicolumn{1}{l}{} & email-univ & ER $(p=0.008)$ & PA $(k=5)$ \\ \midrule
\multicolumn{1}{l}{Random} & $0.8412$ & $0.8790$ & $0.8778$\\
\multicolumn{1}{l}{Two Neighbors} & $0.8717$ & $0.9443$ & $0.8543$\\
\multicolumn{1}{l}{Two Neighbors + High Value} & $0.9161$ & $0.9633$ & $0.9327$\\
\bottomrule
\end{tabular}
\end{minipage}
\end{small}
\end{table*}

\smallskip\noindent\textbf{Real-world Network}
We run simulations using a dataset of an email network consisting of $1133$ nodes and $5451$ edges~\cite{nr}. The network has a density of $0.0085$, an average degree of $9$, and a maximum degree of $71$. 
We also compare these simulations against an {\ER} model 
and a PA model with the same number of vertices and density.
From Table \ref{tab:6}, both \emph{two neighbors} and \emph{two neighbors + high value} ordering outperform the random ordering in the email-univ network. Furthermore, from Figure \ref{fig:real world}, the two neighbors + high value ordering performs the best for the email-univ network, with a median learning rate close to $1$. While both orderings perform better on average, they are both susceptible to negative information cascades. The shape of the email-univ learning rate distribution resembles the PA model in its polarity of outlier data points. However, the outliers for the email-univ network appear to be more evenly distributed compared to the PA model.

\begin{figure}[htbp]
\centering
\vspace*{-4mm}
\includegraphics[scale = 0.5]{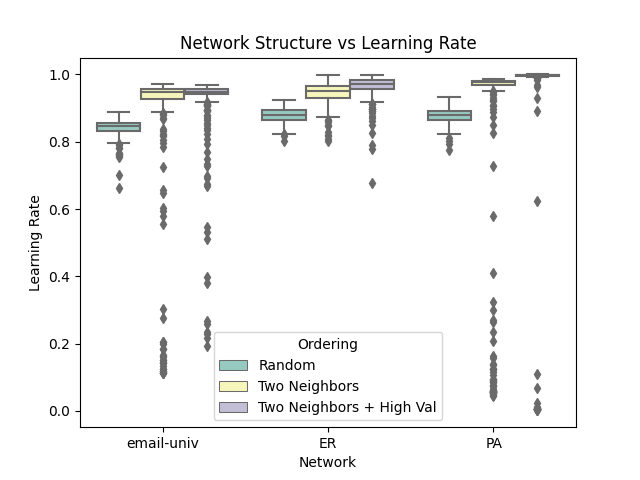}
\caption{Simulation results for the email-univ network with $q=0.7$ over $300$ iterations. }
\label{fig:real world}
\end{figure}

Additional simulation results on the {\ER} graph, PA model and low degree graphs with different network parameters can be found in~\Cref{sec:simulation}.

\section{Conclusions}
Among all the graphs considered in this paper, except a few cases of very sparse or very dense graphs, the quality of learning in a network can vary drastically with different decision orderings. Given an input graph we ask if there exists a decision ordering that supports asymptotic truth learning. We conjecture that this problem is NP-hard.
Furthermore, enabling asymptotic truth learning in graphs of constant degree is generally more challenging.
It is also more difficult to analyze the majority vote model although majority vote model is much easier to implement in practice and achieves good performance if the ordering is chosen right. It is future work to further characterize sparse graphs and majority vote model that have asymptotic truth learning.

\appendix
\section{Proofs in \Cref{sec:challenges}}\label{sec:appendix-challenge}

\SparseGraphNoLearning*

\begin{proof}
    We consider a random ordering $\pi$ of the vertices of any graph, $G$ in this family. For a vertex $u$, define $x(u)$ as a Bernoulli random variable with value $1$ if vertex $u$ appears before all its neighbors in the ordering $\pi$. Define $y(u)=\prob{x(u)=1}= 1/(d(u)+1)$, where $d(u)$ is the degree of $u$.
    Now we examine the learning rate of $u$:
    \begin{align*}
        &\ell(u)=\prob{a_u=\theta}\\
        =&\prob{a_u=\theta|x(u)=1}\prob{x(u)=1}+\prob{a_u=\theta|x(u)=0}\prob{x(u)=0}\\
        =&\prob{a_u=\theta|x(u)=1}y(u)+\prob{a_u=\theta|x(u)=0}(1-y(u))\\
        \leq &q\cdot y(u)+1\cdot (1-y(u))=1-(1-q)\cdot y(u)\\
        =& 1-\frac{1-q}{d(u)+1}        
    \end{align*}
    If the average degree in the network is $\Delta$, then at least $n/2$ nodes have degree no greater than $2\Delta$ -- otherwise, there are more than $n/2$ vertices with degree more than $2\Delta$ each, which already gives a total degree more than $n\Delta$, a contradiction. Thus, 
    \begin{align*}
        L(G)\leq & 1- \sum_u \frac{1}{n} \cdot \frac{1-q}{d(u)+1} 
        \leq  1- \sum_{u,\, d(u)\leq 2\Delta}  \frac{1}{n} \cdot \frac{1-q}{d(u)+1} \\
        \leq & 1- \sum_{u,\, d(u)\leq 2\Delta}  \frac{1}{n} \cdot \frac{1-q}{2\Delta+1} 
        \leq  1- \frac{1-q}{2(2\Delta+1)} 
    \end{align*}
Since $L(G)$ is constant away from $1$, asymptotic truth learning does not happen.
\end{proof}

\ButterFly*

\begin{proof}
Due to symmetry in the network the learning quality of vertices of the same depth is the same. We denote the learning quality of a vertex on depth $i$ to be $q_i$, $1\leq i\leq k+1$. Vertices of depth one make decisions purely from their individual observation, $q_1=q$.

Consider a node $v$ on depth $i$. $v$ has two neighbors $u, w$ in depth $i-1$, which by design, can be viewed as children of $v$ in a binary tree. The actions of $u, w$ are decided by their individual observations as well as the actions of their children in a bottom-up ordering. The actions of $u, w$ are independent of each other.
Therefore, $v$ learns successfully if 1) both $u, w$ predicted correctly; 2) or when $u, w$ have different predictions but $v$'s own observation is correct. This is true for both the Bayesian model and the majority rule model. Therefore,
\begin{align*}
    q_i &= q_{i-1}^2+q\cdot 2q_{i-1}(1-q_{i-1}) 
    = q_{i-1}[1+(1-q_{i-1})(2q-1)]
\end{align*}
The term $(1-q_{i-1})(2q-1) > 0$ for $1/2 < q < 1$, which means that $q_i > q_{i-1}$ and the learning rate increases monotonically as we move up in depth. 

Solving this recurrence is a bit involved, so instead, we can simply bound $q_i$. Let $q_i = 1-\epsilon_i$. It follows that $\epsilon_{i} < \epsilon_{i-1} < 1-q$. Plugging in and rearranging, we get
\begin{align*}
    \epsilon_i &= (2q-1)\epsilon_{i-1}^2+2(1-q)\epsilon_{i-1} < (2q-1)(1-q)\epsilon_{i-1}+2(1-q)\epsilon_{i-1}\\
    &
    = (2q+1)(1-q)\epsilon_{i-1}
\end{align*}
One can verify that $0 < (2q+1)(1-q) < 1$. As a result, $\epsilon_i$ will shrink exponentially fast as $i$ grows: $$\epsilon_i < [(2q+1)(1-q)]^{i-1}\epsilon_1=[(2q+1)(1-q)]^{i-1}(1-q).$$ 

Now we can calculate the network learning rate:
\begin{align*}
L(G) &= \frac{1}{k+1}\sum^{k+1}_{i=1}q_i \geq \frac{1}{1+ k} \sum^k_{i=0} [1-((2q+1)(1-q))^i(1-q)] \\
&= 1 - \frac{1-q}{1+k} \sum^k_{i=0} [(2q+1)(1-q)]^i \\
& =1 - \frac{1-q}{1+k} \cdot \frac{1-[(2q+1)(1-q)]^{k+1}}{1-(2q+1)(1-q)}
\end{align*}
When $n\to \infty$, $k\to \infty$, $L(G)\to 1$.
\end{proof}

\section{Proofs in \Cref{sec:sufficient}}\label{appendix:proof-sufficient}


\BayesianRate*

To prove this proposition rigorously, we first prove a lemma. To do so, we introduce some notation to simplify the following section. Let the vector of random variables, $J_k = (p_1, p_2, \cdots, p_k)$ be the private signals of a learning network $(G, q, \sigma)$ up to $v_k$ in the ordering. Then let $I \in \{0, 1\}^k$ be a vector of binary values that $J$ can possibly take on. Similarly, define $B_k$ as the vector of random variables $(a_i| v_i \in N_\sigma(v_k), p_k)$. $B_k$ here represents all information $v_k$ has access to before making its decision $a_k$. Let $A \in \{0, 1\}^{|N_\sigma(v_k)| + 1}$ be a vector of binary values that $B_k$ can possibly take on. Note that the value of $J_k$ completely determines the value of $B_k$, which in turn determines $a_k$. In other words, if $J_k = I$, then there exists some $A$ such that $\mathbb{P}(B_k = A|J_k = I) = 1$ and $\mathbb{P}(a_k = C|B_k = A) = 1$ for some $C \in \{0,1\}$.
We also define $I^c$ to be $I$ with every value flipped, similarly for $A$. Lastly, let $\mathbb{P}(J = I|\truth = C)$ mean the probability that the random variables in $J$ attain the values in $I$ given that the probability of a private signal giving $C$ is $q$ and $1 - C$ is $1-q$.

\begin{restatable}{lemma}{BayesianComplement}\label{lemma:flip}
Given any learning network $(G, q, \sigma)$ under the Bayesian model, $J_k = I$ determines $a_{k} = 0$, if and only if $J_k = I^c$ determines $a_{k} = 1$. The same is true for $B_k$, if $B_k = A$ determines $a_k = 0$, then $B_k = A^c$ determines $a_k = 1$.
\end{restatable}
\begin{proof}
First, we just need to prove the statement for $J$ since every value of $B_k$ is induced by some value of $J$. Now, we can prove the statement of \Cref{lemma:flip} by induction on $k$. The base case, $k=1$ is trivial. Now, let's assume it's true for $1,...k-1$ and look at $v_k$. First, if $|N_{\sigma}(v_{k})| = 0$, this is almost the same as the base case. Clearly $v_{k}$ just adopts whatever its private signal gives it, so for any $I$, $I^c$ always flips the $k$th private signal which would flip $a_{k}$. If $N_{\sigma}(v_{k})$ is not empty, we first know that because $J = I$, $B_k = A$ and $a_k = C$ for some $A$ and $C$. If $C = 0$, by definition we have:
\[\mathbb{P}(\truth = 0|B_k = A) > \mathbb{P}(\truth = 1|B_k = A)\]
Using Bayes' Rule and the fact that $\theta$ is chosen from $\{0,1\}$ with equal probability:
\[\mathbb{P}(B_k = A|\truth = 0) > \mathbb{P}(B_k = A|\truth = 1)\]
$0$ and $1$ are just arbitrary labels, so if we swap them entirely in the inequality, the value remains unchanged:
\[\mathbb{P}(B_k = A^c|\truth = 1) > \mathbb{P}(B_k = A^c|\truth = 0)\]
Now we can use the induction hypothesis to claim $J^c$ induces $B_k = A^c$ for all elements in $N_{\sigma}(v_{k})$. Then, the inequality above implies given $J^c$, $a_{k} = 1$.
\end{proof}
\begin{proof}[Proof of \Cref{proposition:propagate}]
Let $v_\beta$ be the vertex that achieves max learning rate in $N_{\sigma}(v_\alpha)$. We can write:
\[\ell(v_\alpha) = \ell(v_\beta) - \mathbb{P}(a_\beta = \theta, a_\alpha = 1-\theta) + \mathbb{P}(a_\beta = 1- \theta, a_\alpha = \theta)\]
Let $O$ be a set of vectors, $A$, such that if $B_\alpha = A$, then $a_\beta = 1-\theta$ and $a_\alpha = \theta$, so clearly:
\[\mathbb{P}(B_\alpha = A, A \in O) = \mathbb{P}(a_\beta = 1-\theta, a_\alpha = \theta)\]
Then by \Cref{lemma:flip}, we have that if $O^c = \{A^c|A \in O\}$:
\[\mathbb{P}(B_\alpha = A, A \in O^c) = \mathbb{P}(a_\beta = \theta, a_\alpha = 1-\theta)\]
Now let's rewrite $\mathbb{P}(B_\alpha = A, A \in O)$, the same can be applied for $\mathbb{P}(B_\alpha = A, A \in O^c)$:
\[\mathbb{P}(B_\alpha = A, A \in O) = \sum_{A \in O}\mathbb{P}(B_\alpha=A) = \sum_{A \in O}\mathbb{P}(B_\alpha=A|\truth=\theta)\]
Then, for each $A \in O$, we have from the fact that it induces $a_\alpha = \theta$:
\[\mathbb{P}(\truth = \theta|B_\alpha=A) > \mathbb{P}(\truth = 1-\theta|B_\alpha=A)\]
Again, using Bayes' rule and the fact that $\theta$ is chosen from $\{0,1\}$ with equal probability:
\[\mathbb{P}(B_\alpha=A|\truth =\theta) > \mathbb{P}(B_\alpha=A|\truth = 1-\theta)\]
Next, we can flip all $0$ and $1$ for the right side:
\[\mathbb{P}(B_\alpha = A|\truth = \theta) > \mathbb{P}(B_\alpha=A^c|\truth=\theta)\]
Then summing over $A \in O$:
\[\sum_{A \in O}\mathbb{P}(B_\alpha=A) > \sum_{A \in O}\mathbb{P}(B_\alpha = A^c)\]
\[\mathbb{P}(a_\beta = 1- \theta, a_\alpha = \theta) > \mathbb{P}(a_\beta = \theta, a_\alpha = 1-\theta)\]
Therefore, $\ell(v_\alpha) \geq \ell(v_\beta)=\max \{\ell(v_i)| v_i \in N_{\sigma}(v_\alpha)\}$.
\end{proof}

\Erdoslowdegree*

We first state the known mean and variance of the number of isolated vertices in this class of {\ER} graphs and then use that to prove the above theorem.
\begin{proposition}[\cite{Hofstad_2016}, Proposition 5.9]\label{proposition:isolated}
    If $\frac{1}{2n} \leq p \leq \frac{1}{2}$, then the number of isolated vertices in $G(n, p)$ is a random variable, $Y$ with the following expectation and variance:
    \begin{equation}
        \mathbb{E}(Y) = ne^{-pn}(1+\Theta(p^2n))
    \end{equation}
    \begin{equation}
        \Var(Y) \leq \mathbb{E}(Y) + \frac{pn}{n-pn}\mathbb{E}(Y)^2
    \end{equation}
\end{proposition}

\begin{proof}[Proof of \Cref{proposition:ER-low}]
We will show that, with high probability, there does not exist a decision ordering for each graph in this family. 
First, consider when $p(n) \geq \frac{1}{2n}$. Then, we use Proposition~\ref{proposition:isolated} and let $Y$ be the random variable representing the number of isolated vertices in $G$. The proposition tells us with $pn = O(1)$, $\mathbb{E}(Y) = \Theta(n)$, $\Var(Y) = O(n)$. Then, there exists a constant $C$ such that $C < \mathbb{E}(Y)/n$ and $C < 1- \mathbb{E}(Y)/n$. Using Chebyshev's Inequality, we get:
\begin{equation}
    \mathbb{P}(|Y - \mathbb{E}(Y)| \geq Cn) \leq \frac{\Var(Y)}{(Cn)^2} = O(\frac{1}{n})
\end{equation}
This means that $Y = \Theta(n)$ with high probability. Each isolated vertex has learning rate $q$, and since there are $\Theta(n)$ of them with high probability, we have the error rate must be at least $K(1-q)$ for some constant $K$ regardless of the ordering. Thus, the learning rate cannot approach 1 and truth learning is not achieved. For the case that $p < \frac{1}{2n}$, the probability that $Y = \Theta(n)$ can only be higher, so we have the same result.
\end{proof}

\ErdosLargeComponent*

To prove this, we'll use a theorem from \cite{Hofstad_2016}:

\begin{theorem}[\cite{Hofstad_2016}, Theorem 4.8]\label{theorem:giant-component}
    Fix $p > \frac{1}{n}$ for an {\ER} graph of size $n$. Then, for every $\nu \in (\frac{1}{2}, 1)$, there exists $\delta = \delta(\nu, p) > 0$ such that:
    \begin{equation}
        \mathbb{P}(|\mu - \eta_p| \geq n^{\nu-1}) = O(n^{-\delta})
    \end{equation}
    where $\mu$ is the proportion of vertices not in the giant component and $\eta_p$ satisfies:
    \begin{equation}
        \eta_p = e^{pn(\eta_p - 1)}
    \end{equation}
\end{theorem}

\begin{proof}[Proof of \Cref{lemma:ER-large-component}]
First, it's well known that {\ER} graphs with $p > \frac{\log(n)}{n}$ are connected with high probability, so we may assume otherwise. Clearly, the values of $p$ we are considering satisfy the conditions of \Cref{theorem:giant-component} and we can just pick $\eta = \frac{3}{4}$. Now let's show that as $n \to \infty$, $\eta_p \to 0$. We intend to use contradiction, so let's assume otherwise. First, let's take the partial derivative of both sides with respect to $\eta_p$:

    \[\frac{\partial \eta_p}{\partial \eta_p} = 1 \quad \quad \quad \frac{\partial e^{pn(\eta_p - 1)}}{\partial \eta_p} = pe^{pn(\eta_p - 1)}\]

Now clearly, $1 \geq pe^{p(\eta_p - 1)}$, so as $n$ increases, $pn$ increases since $p = \frac{\omega(1)}{n}$, and that would make $\eta_p$ decrease. Since $\eta_p$ monotonically decreases and $\eta_p > 0$, we have $\eta_p$ must converge to some constant $\epsilon > 0$. However, if we plug this back into the equation and take $n \to \infty$:
\[
    \epsilon = e^{\omega(1)(\epsilon - 1)}
\]
Clearly, the right hand side goes to 0 while the left hand side doesn't, so we have a contradiction. Thus, $\eta_p = o(1)$. Then, we have plugging in $\eta_p$ and $\nu = \frac{3}{4}$ into the theorem:
\begin{equation}
    \mathbb{P}(|\mu - \eta_p| \geq n^{-\frac{1}{4}}) = O(n^{-\delta})
\end{equation}
Since every single term other than $\mu$ has been shown to be $o(1)$, we have with high probability that $\mu = o(1)$ and the size of the giant component is $n - o(n)$.
\end{proof}

\ErdosRenyiConstantP*

\begin{proof}
    We calculate the probability that herding happens on any graph $G$ in this family. Suppose the first two vertices in $\sigma$ with probability $(1-q)^2$ have the wrong observation signal. If every vertex $v_i$ starting from the third one, $i\geq 3$, has at least two neighbors that arrive before $v_i$, regardless of the private signal of $v_i$, with both the Bayesian model and majority rule model, $v_i$ has the wrong prediction. We calculate the probability for this to happen.
    
    For any $3\leq i\leq n$, the probability that $v_i$ connects to at least two vertices that arrive earlier is
    \begin{align*}
p_i &= 1- {i-1 \choose 1} \cdot p (1-p)^{i-2} - (1-p)^{i-1}=1-(1-p)^{i-2}(ip+1-2p).
\end{align*}
Thus, by using $(1-x)(1-y)\geq 1-(x+y)$, when $0<x, y<1$, we have
    \begin{align*}
\prod_{i=3}^{n} p_i &= \prod_{i=3}^{n} \left(1-(1-p)^{i-2}(ip+1-2p)\right)\\
&\geq \left(\prod_{i=3}^{k-1} (1-(1-p)^{i-2}(ip+1-2p))\right)\cdot \left(1-\sum_{i=k}^n (1-p)^{i-2}(ip+1-2p)\right)
\end{align*}
We take $n\to\infty$,
    \begin{align*}
\prod_{i=3}^{n} p_i &\geq \left(\prod_{i=3}^{k-1} (1-(1-p)^{i-2}(ip+1-2p))\right)\cdot \left(1-(1-p)^{k-2}(\frac{2}{p}+k-3)\right)
\end{align*}
We take $k\geq 3$ to be a constant such that $0<(1-p)^{k-2}(\frac{2}{p}+k-3)<1$. Since $(1-p)^{k-2}$ decreases much faster than the increase of $\frac{2}{p}+k-3$ as $k$ increases, the inequality can be achieved with a constant $k$. In that case we have $\prod_{i=3}^{n} p_i \geq q'$ with $q'$ as a constant in $(0, 1)$.
Therefore with probability at least $(1-q)^2q'$ herding happens. That is, there is no asymptotic truth learning.
\end{proof}

\ErdosRenyiComplementGreater*
To prove this, we will first prove the following lemma:
\begin{lemma}\label{lemma:large-clique}
A family of {\ER} graphs with parameter $p$ satisfying $p = \frac{\omega(n^\epsilon)}{n}$ for any $1 > \epsilon > 0$ contains with high probability a clique of size $\omega(1)$.
\end{lemma}
\begin{proof}
This proof just re-analyzes some of the complexity in the proof given in~\cite{Grimmett_McDiarmid_1975} to prove this same statement for constant $p$. The proof is mostly re-iterated here for completeness. Let $K_n$ be a random variable for max clique size and $N_n^k$ be the random variable for the number of $k$-cliques for an {\ER} graph of size $n$. First, the following are true:
\begin{equation}\label{equation:gm-lemma}
    \frac{\mathbb{E}(N^k_n)}{\mathbb{E}(N^k_n|A)} \leq \mathbb{P}(K_n \geq k) \leq \mathbb{E}(N^k_n),
\end{equation}
where $A$ is the event that the first $k$ vertices inserted into the {\ER} graph form a complete graph. Here,
\begin{equation}\label{equation:expected-clique}
    \mathbb{E}(N^k_n) = \binom{n}{k}p^{\frac{1}{2}k(k-1)}
\end{equation}
\begin{equation}\label{equation:expected-clique-conditional}
\mathbb{E}(N^k_n|A) = \sum_{i=m}^k \binom{k}{i}\binom{n-k}{k-i}p^{\frac{1}{2}k(k-1)-\frac{1}{2}i(i-1)}
\end{equation}
where $m = \max(0, \{k-\frac{1}{2}n\})$.
The first inequality of \Cref{equation:gm-lemma} can be seen easily through the following:
\[ \mathbb{E}(N^k_n) = \mathbb{E}(N^k_n|K_n \geq k) \mathbb{P}(K_n \geq k) \leq {\mathbb{E}(N^k_n|A)} \mathbb{P}(K_n \geq k) \]

Now define:
\[d(n) = \frac{\log(n)}{\log(1/p)} \quad (n=1,2,...)\]
\[T_i(n) = \frac{1}{\binom{n}{d(n)}}\binom{d(n)}{i}\binom{n-d(n)}{d(n)-i}p^{-\frac{1}{2}i(i-1)} \quad (n = 1,2,...; i = 1,2,...d(n))\]
Note that $d(n) = O(\log(n))$ and $d(n) = \omega(1)$.
We have from \Cref{equation:gm-lemma}:
\[\mathbb{P}(K_n \leq d(n)) \leq 1- \frac{\mathbb{E}(N^k_n)}{\mathbb{E}(N^k_n|A)}\]
Then from \Cref{equation:expected-clique}, and \Cref{equation:expected-clique-conditional}:
\[\mathbb{P}(K_n \leq d(n))\leq 1-(\sum_{i=0}^{d(n)}T_i(n))^{-1} \leq \sum_{i=0}^{d(n)}T_i(n) - 1\]
The second inequality comes from the fact that:
\[\frac{\mathbb{E}(N^k_n)}{\mathbb{E}(N^k_n|A)} \leq 1\]

We can now analyze the complexity of this probability. We have:
\[T_0(n) = (1-\frac{d}{n})(1-\frac{d}{n-1})...(1-\frac{d}{n-d+1}) \leq 1\]
\[T_1(n) = T_0(n)\frac{d^2}{n-2d+1} = O(\frac{\log^2(n)}{n})\]
Then, we have clearly $T_0 + T_1 - 1 = O(\frac{\log^2(n)}{n})$. Next, to analyze $2 \leq j \leq d(n)$:
\[\frac{T_j(n)}{T_2(n)} = \frac{(n-2d+2)!}{(n-2d+j)!}(\frac{(d-2)!}{(d-j)!})^2\frac{2p^{1-\frac{1}{2}j(j-1)}}{j!} \leq (\frac{d^2}{n-2d}p^{-\frac{1}{2}(j+1)})^{j-2}\]
Using $j \leq d$ and $d = \frac{\log(n)}{\log(1/p)}$:
\[\frac{T_j(n)}{T_2(n)} \leq (\frac{d^2}{n-2d}(1/p)^{\frac{\log(n)}{2\log(1/p)}+\frac{1}{2}})^{j-2} \leq (\frac{d^2(n^{1/2}(1/p)^{1/2}}{n-2d})^{j-2} = o(1)\]
Now we find $T_2(n)$:
\[T_2(n) \leq p^{-2}\frac{(d-2)(d-1)}{2(n-2d+1)}T_1(n) = O(\frac{\log^4(n)}{n^2})o(n^{2\epsilon})\]
Plugging everything back, gives us:
\[P(K_n < d(n)) \leq T_0(n) + T_1(n) - 1 + d(n)T_2(n) = o(1)\]
Thus, with high probability, there exists a clique of size at least $d(n) = \omega(1)$ in $G(n, p)$.
\end{proof}
\begin{proof}[Proof of \Cref{prop:ER-complement-greater}]
First, we can look at $G' = G(\frac{n}{2}, p)$. We have that $1-p = \frac{\omega((\frac{n}{2})^\epsilon)}{(\frac{n}{2})}$, so we can use \Cref{lemma:large-clique} to claim that there exists a clique of size $\omega(1)$ in the complement of $G'$ with high probability. That implies $G'$ has an independent set, $S$, of that size. Now, we generate the rest of $G(n, p)$ by adding the rest of the $\frac{n}{2}$ vertices and corresponding edges. It is easy to see after picking a random new vertex $v$,  with high probability it connects with at least $\omega(1)$ of the independent set $S$ in $G'$. It is also clear to see that after removing that independent set from $G$ that $G$ remains connected since none of vertices generated in the second half were removed. Thus, all conditions of \Cref{theorem:ordering_exists} are satisfied.
\end{proof}


\GridGraphLearning*

\begin{proof}
Let $G$ be a grid graph of the lattice from $(1,1)$ to $(\sqrt{n},\sqrt{n})$. The following is a decision ordering that achieves asymptotic truth learning. See also \Cref{fig:grid-learning}.
\begin{enumerate}
    \item All node sequences of the following form in order of increasing $a$, starting from $1$ with maximum value in the order of $O(\log \log n)$:
    \[(m-2^{a-1}, a),..., (m - 1, a), (m+2^{a-1}, a),..., (m + 1, a), (m, a)\]
    for all non-negative integer tuples $(a, h)$ satisfying:
    \[k = \max\{x: 2^{x + 1 + a} \leq \floor{\log(n)}\} \quad \quad \quad 1 \leq h \leq 2^{k} \quad \quad \quad m = (2h-1)2^{a}\]
    Notice for a fixed value of $a$, $k$ is decided but there are multiple values of $h$ and $m$. Also if we fix both $a$ and $h$, $m$ is decided. One can verify\footnote{The vertices appearing in the sequence for tuple $(a, h)$ have $y$-coordinate of $a$ and $x$-coordinates within the range of $(4h-3)2^{a-1}$ and $(4h-1)2^{a-1}$.} that one vertex cannot appear in two sequences for different tuples $(a, h)\neq (a', h')$. All vertices used in this step have $x$-coordinate at most $\log n$ and $y$-coordinate at most $\log \log n$.
    We just need all the sequences follow an increasing order of $a$. With the same value of $a$, the sequences with different values of $h$ can be ordered arbitrarily. 
    \item A traversal of $G - (A - \{\alpha\})$ starting at $\alpha$, where $A$ is the set of all nodes from the first step and $\alpha$ is the last node in the first step. 
    \item The rest of the nodes in $G$ in any order.
\end{enumerate}
We will now show that asymptotic truth learning is achieved by aggregation in the first step, and propagation in the second step.

\begin{figure}[htbp]
\centering
\includegraphics[width=.85\linewidth]{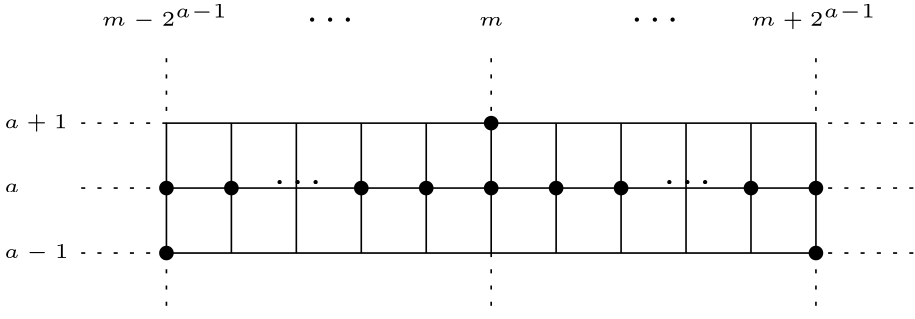}
\caption{Here is an example of a single vertex sequence corresponding to one tuple $(a, h)$ from step 1. The $a-1$ row vertex sequences end at the marked nodes. Then, the vertex sequence in row $a$ starts from the far left and goes to the right one by one until the node $(m-1, a)$. Then, it starts from the far right (node $(m+2^{a-1}, a)$) and walks to the left one by one to the node $(m+1, a)$. Lastly, $(m, a)$ gets to aggregate the signals from $(m-1, a)$ and $(m+1, a)$, which is eventually picked up by vertex sequence in row $a+1$. For a fixed $a$, there are $2^k$ such vertex sequences.}
\label{fig:grid-learning}
\end{figure}

First, we show that node $\alpha$ achieves a learning rate of $1 - o(1)$ in the first step. 
As $a$ increases, the learning rate of nodes in row $a$ increases, and we will show that the learning rate approaches 1.
Consider the sequence for a tuple $(a, h)$ (see \Cref{fig:grid-learning}), by \Cref{proposition:propagate}, we see that $\ell(m-1, a) \geq \ell(m-2^{a-1}, a)$ and $\ell(m + 1, a) \geq \ell(m+2^{a-1}, a)$.  Next, the vertex $(m, a)$ can aggregate signals from both $\ell(m-1, a)$ and $\ell(m-1, a)$. 
Further, since the sequence with $a$ comes after the sequence with $a-1$, we have $\ell(m-2^{a-1}, a)\geq \ell(m-2^{a-1}, a-1)$ and $\ell(m+2^{a-1}, a)\geq \ell(m+2^{a-1}, a-1)$. 
Now we consider the abstraction of the ordering in step 1 as a fully balanced binary tree $T$ with vertices $(m, a)$ mapping to the parent of the vertices corresponding to $(m-2^{a-1}, a-1)$ and $(m+2^{a-1}, a-1)$. Here $a$ corresponds to the levels of the binary tree with the bottom level as $a=1$.  
We aggregate information bottom up with $a$ increasing from $1$ until $\Theta(\log(\floor{\log(n)}))$. Since the learning rate of the root of a fully balanced binary tree with a total of $g$ levels, using a bottom up ordering, is $1-\Theta(e^{-g})$, we have the error rate of $\alpha$ is $O(\frac{1}{\floor{\log(n)}})$, which means its learning rate is $1 - o(1)$.

Next, we have from \Cref{proposition:propagate} that $\alpha$ passes its learning rate to the rest of the graph that can be traversed. This is at least $G - \floor{\log(n)} \times \floor{\log(n)}$, which is $1 - o(1)$ proportion of all of the nodes. Thus, the network learning rate is at least:
\[(1- o(1))^2 = 1 - o(1).\]
This finishes the proof.
\end{proof}

\section{Simulations}\label{sec:simulation}
We will now provide additional simulation results with {\ER} model, the PA model and low degree networks.


\subsection{{\ER} Model}
For {\ER} model, we are interested in understanding the influence of edge probability $p$ on the learning rate of the network. As $p$ increases, the overall connectivity of the network increases, and the network undergoes a transition from a sparsely connected to a more densely connected regime. 


\begin{table*}[h]
\caption{Erdos-Renyi Two Neighbors vs Two Neighbors + High Value}
\label{tab:1}
\begin{small}
\begin{minipage}{\columnwidth}
\centering
\begin{tabular}{lccccc} \toprule
\multicolumn{6}{l}{Simulation results $q=0.7$, $n=1000$, iterations$=300$} \\ \cmidrule(l){1-6}
\multicolumn{1}{l}{\multirow{2}{*}{Ordering}} & \multicolumn{5}{c}{Average learning rate} \\ \cmidrule(l){2-6}
\multicolumn{1}{l}{} & $p=0.005$ & $p=0.01$ & $p=0.03$ & $p=0.05$ & $p=0.07$ \\ \midrule
\multicolumn{1}{l}{Random} & $0.8017$ & $0.8820$ & $0.9612$ & $0.9753$ & $0.9757$ \\
\multicolumn{1}{l}{Two Neighbors} & $0.9079$ & $0.9466$ & $0.9714$ & $0.9747$ & $0.9816$ \\
\multicolumn{1}{l}{Two Neighbors + High Value} & $0.7881$ & $0.9524$ & $0.9734$ & $0.9774$ & $0.9804$ \\ \bottomrule
\end{tabular}
\end{minipage}
\end{small}
\end{table*}

\begin{table*}[h]
\caption{Erdos-Renyi Two Neighbors vs Two Neighbors + High Value}
\label{tab:2}
\begin{small}
\begin{minipage}{\columnwidth}
\centering
\begin{tabular}{lcccc} \toprule
\multicolumn{5}{l}{Simulation results $q=0.7$, $n=1000$, iterations$=300$} \\ \cmidrule(l){1-5}
\multicolumn{1}{l}{\multirow{2}{*}{Ordering}} & \multicolumn{4}{c}{Average learning rate} \\ \cmidrule(l){2-5}
\multicolumn{1}{l}{} & $p=0.1$ & $p=0.2$ & $p=0.3$ & $p=0.7$ \\ \midrule
\multicolumn{1}{l}{Random} & $0.9783$ & $0.9542$ & $0.8972$ & $0.8626$\\
\multicolumn{1}{l}{Two Neighbors} & $0.9761$ & $0.9636$ & $0.9403$ & $0.8622$\\
\multicolumn{1}{l}{Two Neighbors + High Value} & $0.9691$ & $0.9359$ & $0.9232$ & $0.8853$\\ \bottomrule
\end{tabular}
\end{minipage}
\end{small}
\end{table*}
In Table \ref{tab:1} and \ref{tab:2}, we compare the random, \emph{two neighbors}, and \emph{two neighbors + high value} ordering. Here, we simulate over edge probabilities between $p=0.01$ and $p=0.7$. For a small edge probability $p=0.005$, we see relatively similar overall learning rates for all orderings, suggesting that successful aggregation cannot occur in sparsely connected regimes. In general, the \emph{two neighbors + high value} ordering performs better than the \emph{two neighbors} as it prioritizes high-value information. However, as $p$ increases to above $0.05$, we see relatively interchangeable learning rates for all of the different orderings. In more densely connected regimes as seen in Table \ref{tab:2}, both the \emph{two neighbors} and \emph{two neighbors + high value} orderings become more susceptible to negative information cascades. The edge probability $p=0.7$ best illustrates this point as all learning rates fall below $0.9$.

\begin{figure}[h]
\centering
\includegraphics[scale = 0.4]{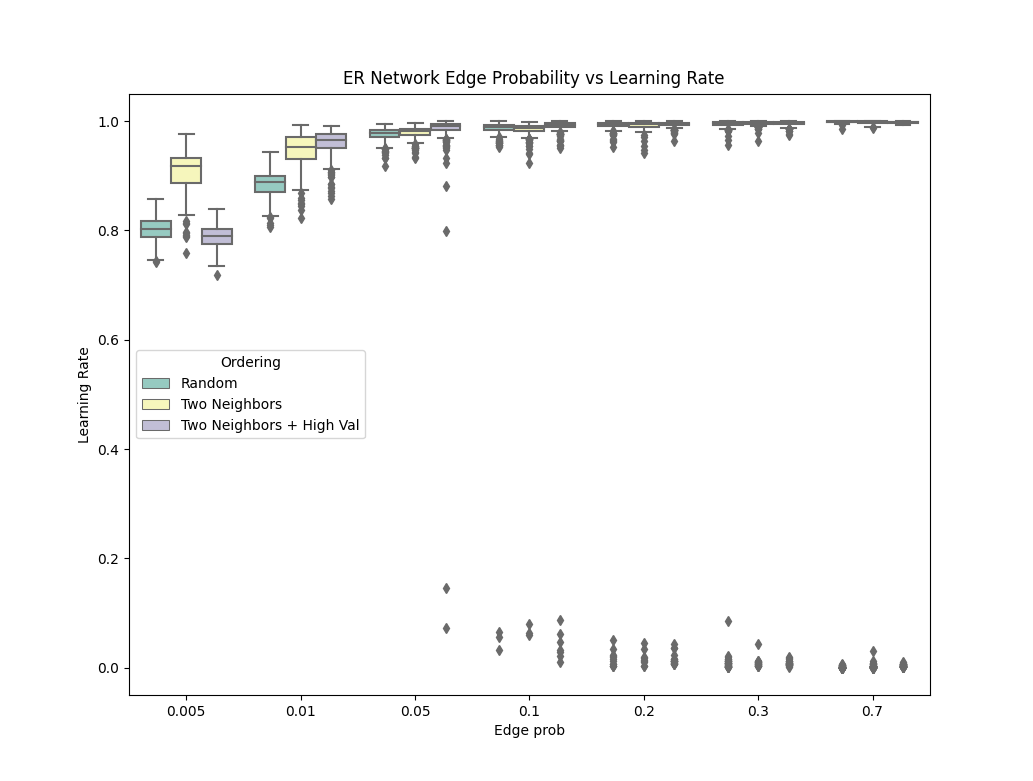}
\vspace*{-4mm}
\caption{\small The simulation results for an ER network with $q=0.7$ and $n=1000$, over $300$ iterations. The median learning rate over 300 iterations increases as the edge probability increases. However, the number of negative information cascades also increases, which is reflected in the polarity of learning as well as the lower mean learning rates.}
\label{fig:ER Model}
\end{figure}

\subsection{Preferential Attachment Model}
In preferential attachment (PA) graphs, a key observation is that there will be a few extremely high-degree nodes. This allows us to utilize the \emph{two neighbors + high value} ordering with a high independent set size $m$ which will allow us to minimize the chance of a negative information cascade. In our simulations, we utilize a PA model with 
a seed clique of $k = 5$, and we attach each new node to $k$ preexisting nodes. Additionally, we will consider three orderings: 1.) a random ordering, 2.) the \emph{arrival} ordering of the PA model, and 3.) and the \emph{Two Neighbor + High Value} ordering with parameter $m = 30$. 
Over the three orderings, we observe the behavior network learning rate as $n$ increases. In Table \ref{tab:3}, we observe that the \emph{two neighbor + high value} ordering consistently performs the best, with a learning rate increasing with $n$. The arrival ordering is extremely susceptible to information cascades. 
This causes high variation in learning rate we see in Table \ref{tab:3}. 
\begin{table*}[h]
\caption{Preferential Attachment with Two Neighbors + High Value}
\label{tab:3}
\begin{small}
\begin{minipage}{\columnwidth}
\centering
\begin{tabular}{lcccccc} \toprule
\multicolumn{7}{l}{Simulation results $q=0.7$, $k = 5$, iterations $=300$} \\ \cmidrule(l){1-7}
\multicolumn{1}{l}{\multirow{2}{*}{Ordering}} & \multicolumn{6}{c}{Average learning rate} \\ \cmidrule(l){2-7}
\multicolumn{1}{l}{} & $n = 500$ & $n=700$ & $n=900$ & $n=1100$ & $n=1300$ & $n = 1500$\\ \midrule
\multicolumn{1}{l}{Random} & $0.8813$ & $0.8800$ & $0.8791$ & $0.8821$ & $0.8816$ & $0.8799$\\
\multicolumn{1}{l}{Arrival} & $0.8845$ & $0.8257$ & $0.8606$ & $0.8254$ & $0.8935$ & $0.8312$\\
\multicolumn{1}{l}{High Value} & $0.9577$ & $0.9689$ & $0.9758$ & $0.9832$ & $0.9842$ & $0.9851$ \\ 
\bottomrule
\end{tabular}
\end{minipage}
\end{small}
\end{table*}

\subsection{Low Degree Structures}
We define the $n \times n$ grid $G(n,n)$ as the regular grid of nodes in $\mathbb{R}^2$, where the node set consists of label pairs $(i, j) \in \{1,..., n\}^2$ and an edge set $\{(i,j)(i',j'):|i-i'|+|j-j'|=1\}$. The \textit{boundary} of $G(n, n)$ is the outermost layer of the grid and the interior nodes are all of the rest. We induce a spiral ordering on the nodes. For a spiral order, we begin at the top left corner node of the grid and follow a clockwise spiral until we account for all boundary nodes. We then repeat the same process for the next layer down, starting at the top left corner node and traversing clockwise, like peeling layers of an onion. 

Both the binary butterfly network as seen in \Cref{fig:butterfly} and the grid network are examples of \emph{low degree graphs}, where the maximum degree is bounded by a constant. Unlike {\ER} or PA models, they don't have high-degree nodes, or hubs, that make for easy information aggregation and propagation. However, under specific orderings that exploit the local structure of the network, information can be incrementally aggregated. As seen in Table \ref{tab:4} and \ref{tab:5}, the performance of both networks improves significantly as the $n$ increases, in contrast with the random ordering that appears to have a learning rate asymptotically bounded away from $1$. 
\begin{table*}[htbp]
\caption{Butterfly Network}
\label{tab:4}
\begin{small}
\begin{minipage}{\columnwidth}
\centering
\begin{tabular}{lcccccc} \toprule
\multicolumn{7}{l}{Simulation results $q=0.7$, iterations$=300$} \\ \cmidrule(l){1-7}
\multicolumn{1}{l}{\multirow{2}{*}{Ordering}} & \multicolumn{6}{c}{Average learning rate} \\ \cmidrule(l){2-7}
\multicolumn{1}{l}{} & $n=80$ & $n=192$ & $n=448$ & $n=1024$ & $n=2304$ & $n=5120$ \\ \midrule
\multicolumn{1}{l}{Random} & $0.7602$ & $0.7600$ & $0.7634$ & $0.7643$ & $0.7640$ & $0.7663$\\
\multicolumn{1}{l}{Bottom-Up} & $0.8426$ & $0.8513$ & $0.8717$ & $0.8859$ & $0.9000$ & $0.9082$\\ \bottomrule
\end{tabular}
\end{minipage}
\end{small}
\end{table*}

\begin{table}[htbp]
\caption{$n \times n$ Grid Network}
\label{tab:5}
\begin{small}
\begin{minipage}{\columnwidth}
\centering
\begin{tabular}{lccccc} \toprule
\multicolumn{6}{l}{Simulation results $q=0.7$, iterations$=300$} \\ \cmidrule(l){1-6}
\multicolumn{1}{l}{\multirow{2}{*}{Ordering}} & \multicolumn{5}{c}{Average learning rate} \\ \cmidrule(l){2-6}
\multicolumn{1}{l}{} & $n=20$ & $n=30$ & $n=40$ & $n=50$ & $n=60$ \\ \midrule
\multicolumn{1}{l}{Random} & $0.7714$ & $0.7718$ & $0.7721$ & $0.7727$ & $0.77305$ \\
\multicolumn{1}{l}{Spiral} & $0.8922$ & $0.9168$ & $0.9366$ & $0.9501$ & $0.9575$ \\ \bottomrule
\end{tabular}
\end{minipage}
\end{small}
\end{table}

\begin{figure}[h]
\centering
\includegraphics[scale = 0.6]{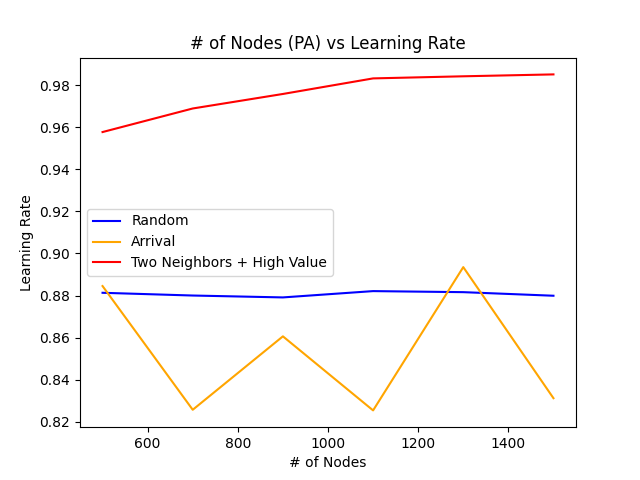}
\caption{The simulation results for a PA model with $q=0.7$ and $k=5$ over $300$ iterations. The two neighbors + high value ordering performs much better than random and arrival orderings.}
\label{fig:PA Model}
\end{figure}

\begin{figure}[h]
\centering
\begin{minipage}{.5\textwidth}
  \centering
  \includegraphics[scale = 0.45]{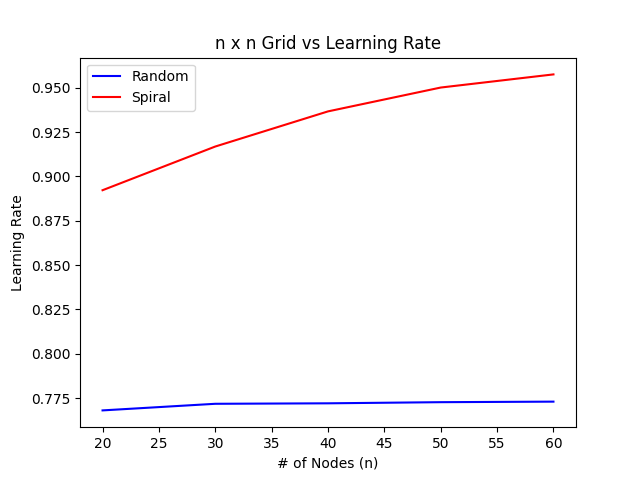}
  \label{fig:grid}
\end{minipage}%
\begin{minipage}{.5\textwidth}
  \centering
  \includegraphics[scale = 0.45]{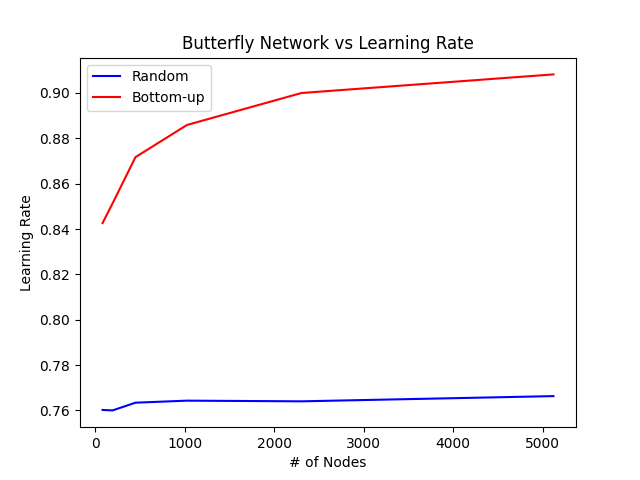}
  \label{fig:butterfly network}
\end{minipage}
\captionof{figure}{The simulation results for low-degree networks with $q=0.7$ over $300$ iterations. We observe the behavior of low-degree networks, the $n \times n$ grid network (left) and butterfly network (right), as the number of nodes increases under specific orderings. In both cases, we see the feasibility of asymptotic learning.}
\end{figure}

\end{document}